\documentclass[12pt]{article}

\usepackage[dvips]{epsfig}
\usepackage{latexsym}
\usepackage{color}

\usepackage{subfig}

\def\a{\alpha}
\def\b{\beta}

\newcommand{\nba}{{\sc nba}}

 \newtheorem{thm}{Theorem}[section]

 \newtheorem{lemma}[thm]{Lemma}
\newtheorem{cor}[thm]{Corollary}

\newtheorem{claim}{Claim}[section]

\newcommand{\ufp}{{\sc ufp}}

\newcommand{\ssufp}{Single-Sink {\sc max-ufp}}

\newcommand{\qed}{$\hfill{\Box}$}

\newcommand{\medp}{{\sc MEDP}}
\newcommand{\congestionresults}[1]{}
\newcommand{\capp}[1]{c_{#1}}

\begin{document}
\bibliographystyle{plain}


\title{The Inapproximability of Maximum Single-Sink\\ Unsplittable, Priority and Confluent Flow Problems}
\author{F. Bruce Shepherd \and Adrian Vetta}

\maketitle

\begin{abstract}
We consider the  single-sink network flow problem.
An instance consists of a capacitated graph (directed or undirected),
a sink node $t$ and a set of demands that we want to send to the
sink. Here demand $i$ is located at a node $s_i$ and
requests an amount $d_i$ of flow capacity in order to route successfully.
Two standard objectives are to maximise
(i) the number of demands (cardinality) and (ii) the total demand (throughput)
that can be routed subject to the capacity constraints. Furthermore, we examine these maximisation
problems for three specialised types of network flow: unsplittable, confluent and priority flows.

In the {\em unsplittable flow} problem, we have edge capacities, and the demand for $s_i$ must be
routed on a single path. In the {\em confluent flow} problem, we have node capacities, and the final
flow must induce a tree.
Both of these problems have been studied extensively, primarily in the
single-sink setting. However, most of this work imposed the {\em no-bottleneck assumption}
(that the maximum demand $d_{max}$ is at most the minimum capacity $u_{min}$).
Given the no-bottleneck assumption, there is a factor $4.43$-approximation algorithm due to Dinitz et al.~\cite{Dinitz99}
for the unsplittable flow problem. 
Under the even stronger assumption of uniform capacities,
there is a factor $3$-approximation algorithm due to Chen et al.~\cite{Chen07} for the confluent flow problem.
However, unlike the unsplittable flow problem, a constant factor approximation algorithm cannot
be obtained for the single-sink confluent flow problem even {\bf with} the no-bottleneck assumption.
Specifically, we prove that it is hard in that setting to approximate single-sink confluent flow to
within $O(\log^{1-\epsilon}(n))$, for any $\epsilon>0$.
This result applies for both cardinality and throughput objectives even in undirected graphs.

The remainder of our results focus upon the setting {\bf without} the no-bottleneck assumption.
There, the only result we are aware of is an $\Omega(m^{1-\epsilon})$ inapproximability
result of Azar and Regev~\cite{azar2001strongly} for cardinality single-sink unsplittable flow in directed graphs.
We prove this lower bound applies to undirected graphs, including planar networks.
This is the first super-constant hardness known for undirected single-sink unsplittable flow,
and apparently the first polynomial hardness for undirected unsplittable flow even for general (non-single sink)
multiflows.
We show the lower bound also applies to the cardinality single-sink confluent flow
problem.

Furthermore, the proof of Azar and Regev requires exponentially large demands.
We show that polynomial hardness continues to hold without this restriction, even if all demands
and capacities lie within an arbitrarily small range $[1,1+\Delta]$, for $\Delta > 0$. This lower bound applies
also to the throughput objective.
This result is very sharp since if $\Delta=0$,  then we have an instance of the single-sink maximum
edge-disjoint paths problem which can be solved exactly via a maximum flow algorithm.
This motivates us to study an intermediary problem, {\em priority flows}, that models
the transition as $\Delta\rightarrow 0$. Here we have unit-demands, each with a priority level.
In addition, each edge has a priority level and a routing path for a
demand is then restricted to use edges with at least the same priority level.
Our results imply a polynomial lower bound for the maximum priority flow problem, even for
the case of uniform capacities.

Finally, we present greedy algorithms that provide upper bounds which (nearly) match
the lower bounds for unsplittable and priority flows. These upper bounds also apply for
general multiflows.
\end{abstract}

\section{Introduction}
In this paper we improve known lower bounds (and upper bounds) on the approximability
of the maximization versions of the {\em single-sink unsplittable flow}, {\em single-sink priority flow}
and {\em single-sink confluent flow} problems.
In the single-sink network flow problem, we are given
a directed or undirected graph $G=(V,E)$ with $n$ nodes and $m$ edges that has edge capacities $u(e)$ or
node capacities $u(v)$. There are a collection of demands
that have to be routed to a unique destination {\em sink} node $t$.
Each demand $i$ is located at a {\em source} node $s_i$ (multiple demands could share the same source)
and requests an amount $d_i$ of flow capacity in order to route.
We will primarily focus on the following two well-known versions of the single-sink network flow problem:
\begin{itemize}
\item {\tt Unsplittable Flow}: Each demand $i$ must be sent along a unique path $P_i$ from
$s_i$ to $t_i$.
\item {\tt Confluent Flow}: Any two demands that meet at a node must then traverse identical paths
to the sink.  In particular, at most one edge out of each node $v$ is allowed to carry flow.
Consequently, the support of the flow is a tree in the
undirected graphs, and an arborescence rooted at $t$ in directed graphs.
\end{itemize}
Confluent flows were introduced to study the effects of next-hop routing \cite{Chen05}.
In that application, routers are capacitated and, consequently, nodes
in the confluent flow problem are assumed to have capacities but not edges.
In contrast, in the unsplittable flow problem it is the edges that are assumed to
be capacitated. We follow these conventions in this paper.
In addition, we will also examine a third network flow problem called {\tt Priority Flow} (defined in Section \ref{sec:results}).
In the literature, subject to network capacities, there are two standard maximization objectives:


\begin{itemize}
\item {\tt Cardinality}: Maximize the total number of demands routed.
\item {\tt Throughput}: Maximize satisfied demand, that is, the total flow carried by the routed demands.
\end{itemize}
These objectives can be viewed as special cases of the
{\em profit-maximisation} flow problem. There each demand $i$ has
a profit $\pi_i$ in addition to its demand $d_i$. The goal is to route a subset
of the demands of maximum total profit. The cardinality model then corresponds to the
unit-profit case, $w_i=1$ for every demand $i$; the throughput model is the case $\pi_i=d_i$.
Clearly the lower bounds we will present also apply to the more general
profit-maximisation problem.



\subsection{Previous Work}\label{sec:previous}
The unsplittable flow problem has been extensively studied since its
introduction by Cosares and Saniee \cite{cosares1994optimization} and Kleinberg~\cite{Kleinberg96}. However, most positive results
have relied upon the {\em no-bottleneck assumption} (\nba)
where the maximum demand is at most the minimum
capacity, that is, $d_{max} \leq u_{min}$.
Given the no-bottleneck assumption, the best known result is a factor $4.43$-approximation
algorithm due to Dinitz, Garg and Goemans \cite{Dinitz99} for the maximum throughput objective.


%

The confluent flow problem was first examined by Chen, Rajaraman and Sundaram~\cite{Chen05}.
There, and in variants of the problem \cite{Chen07, donovan2007degree, shepherd2009single},  the focus
was on  uncapacitated graphs.\footnote{An exception concerns the analysis of
graphs with constant treewidth \cite{dressler2010capacitated}.}
The current best result for maximum confluent flow is a factor $3$-approximation
algorithm for maximum  throughput in uncapacitated networks \cite{Chen07}.

Observe that uncapacitated networks  (i.e. graphs with uniform capacities) trivially also satisfy
the no-bottleneck assumption. Much less is known about networks where the
no-bottleneck assumption does {\bf not} hold. This
is  reflected by the dearth of progress for the case of multiflows (that is, multiple sink)
without the \nba.
It is known that a constant factor approximation algorithm exists for the case in which
$G$ is a path \cite{bonsma2011constant}, and that a poly-logarithmic approximation algorithm exists for the
case in which $G$ is a tree \cite{chekuri2009unsplittable}.
The extreme difficulty of the unsplittable flow problem is suggested by the following
result of Azar and Regev~\cite{azar2001strongly}.
\begin{thm}[\cite{azar2001strongly}]\label{thm:ar}
If $P \neq NP$ then, for any $\epsilon > 0$, there is no $O(m^{1-\epsilon})$-approximation algorithm for the cardinality objective
of the single-sink unsplittable flow problem in directed graphs.
\end{thm}
This is the first (and only) super-constant lower bound for the maximum single-sink unsplittable flow problem.



\subsection{Our Results}\label{sec:results}
The main focus of this paper is on single-sink flow problems where the no-bottleneck assumption
does not hold. It turns out that the hardness of approximation bounds are quite severe
even in the (often more tractable) single-sink setting.
In some cases they  match the worst case bounds for PIPs (general packing integer programs).
In particular, we strengthen Theorem~\ref{thm:ar} in four ways.
First, as noted by Azar and Regev, the proof of their result relies critically
on having directed graphs.
We prove it holds for undirected graphs, even {\em planar} undirected graphs.
Second, we show the result also applies to the confluent
flow problem.
\begin{thm}
\label{thm:extended}
If $P \neq NP$ then, for any $\epsilon > 0$, there is no $O(m^{1-\epsilon})$-approximation algorithm
for the cardinality objective of the single-sink unsplittable and confluent flow problems in undirected graphs.
Moreover for unsplittable flows, the lower bound holds even when we restrict to planar inputs.
\end{thm}

Third, Theorems \ref{thm:ar} and \ref{thm:extended} rely
upon the use of exponentially large demands -- we call this the {\em large demand regime}. A second
demand scenario that has received attention in the literature is the {\em polynomial demand regime} -- this regime
is studied in \cite{guruswami2003near}, basically to the exclusion of the large demand regime.
We show that strong hardness results apply in the polynomial demand regime; in fact,
they apply to the {\em small demand regime} where the {\em demand spread}
 $\frac{d_{max}}{d_{min}} = 1+\Delta$, for some ``small'' $\Delta > 0$.
(Note that $d_{min} \leq u_{min}$ and so the demand spread of an instance is at least
the {\em bottleneck value} $\frac{d_{max}}{u_{min}}$.)
Fourth, by considering the case where $\Delta  > 0$ is arbitrarily small we obtain
similar hardness results for the throughput objective for the single-sink unsplittable
and confluent flow problems.
Formally, we show the following $m^{\frac12 - \epsilon}$-inapproximability result. We
note however that the hard instances have a linear number of edges (so one may prefer to
call this an $n^{\frac12-\epsilon}$-inapproximability result).

\begin{thm}
\label{thm:hard} Neither cardinality nor throughput can be approximated to
within a factor of $O(m^{\frac12-\epsilon})$, for any
$\epsilon > 0$, in the single-sink unsplittable and confluent flow problems.
This holds for  undirected and directed graphs even
when instances are restricted to have demand spread
$\frac{d_{max}}{d_{min}}=1 + \Delta$, where $\Delta  > 0$ is arbitrarily small.
\end{thm}
Again for the unsplittable flow problem this hardness result applies even in planar graphs.
Theorems \ref{thm:extended} and \ref{thm:hard} are the first super-constant hardness for any undirected
version of the single-sink unsplittable flow problem, and any directed version with small-demands.
We also remark that the extension to the small-demand regime is significant as suggested by the sharpness of the result.
Specifically, suppose $\Delta=0$ and, thus, the demand spread is one. We may then scale to assume that $d_{max}=d_{min}=1$.
Furthermore, we may then round down all capacities to the nearest integer as any fractional capacity cannot be used.
But then the single-sink unsplittable flow problem can be solved easily in polynomial time by a max-flow algorithm!

To clarify what is happening in the change from $\Delta>0$ to $\Delta=0$, we introduce
and examine an intermediary problem, the {\em maximum priority flow problem}.
Here, we have a graph $G=(V,E)$ with a sink node $t$, and demands from nodes $s_i$ to $t$.
These demand are unit-demands, and thus $\Delta=0$. However, a demand may not traverse every edge.
Specifically, we have a partition of $E$ into priority classes $E_i$. Each demand also has a {\em priority}, and
a demand of priority $i$ may only use edges of priority $i$ or better (i.e., edges in $E_1 \cup E_2 \cup \ldots E_i$).
The goal is to find a maximum routable subset of the demands. Observe that, for this unit-demand problem, the
throughput and cardinality objectives are identical.
Whilst various priority network design problems have been considered in the literature (cf. \cite{charikar2004resource,chuzhoy2008approximability}),
we are not aware of existing results on maximum priority flow. Our results immediately imply the following.
\begin{cor}
\label{cor:priority}
The single-sink maximum priority flow problem
cannot be approximated to within a factor of $m^{\frac12-\epsilon}$, for any $\epsilon > 0$, in planar
directed or undirected graphs.
\end{cor}

The extension of the hardness results for single-sink unsplittable flow to undirected graphs is also significant since
it appears to have been left unnoticed even
for  general multiflow instances.
In \cite{guruswami2003near}: ``...{\em the hardness of undirected edge-disjoint paths remains an
interesting open question. Indeed, even the hardness of
edge-capacitated unsplittable flow remains open''}\footnote{In \cite{guruswami2003near}, they do however establish
an inapproximability bound of $n^{1/2-\epsilon}$, for any $\epsilon > 0$,
on {\em node-capacitated} {\sf USF} in undirected graphs.}  Our result resolves this question by showing
polynomial hardness (even for single-sink instances). We emphasize that this is not the first
super-constant hardness for general multiflows however.
  A polylogarithmic  lower bound appeared in \cite{andrews2006logarithmic} for the maximum
  edge-disjoint paths (MEDP) problem (this was subsequently extended to the regime where edge congestion is
  allowed \cite{andrews2010inapproximability}).
Moreover,  a polynomial lower bound for MEDP seems less likely given the recent $O(1)$-congestion polylog-approximation
algorithms \cite{chuzhoy2012routing,chuzhoy2012polylogarithimic}.
In this light, our hardness results for single-sink
unsplittable flow again highlight the sharp threshold involved with the no-bottleneck assumption.
That is, if we allow some slight
variation in demands and capacities within a tight range $[1,1+\Delta]$  we immediately jump from (likely)
polylogarithmic approximations for MEDP to  (known) polynomial hardness of the corresponding maximum
unsplittable flow instances.

We next note that Theorems \ref{thm:ar} and \ref{thm:extended} are stronger than Theorem \ref{thm:hard}
in the sense that they have exponents of $1-\epsilon$ rather than $\frac12-\epsilon$. Again, this extra boost is due to their
use of exponential demand sizes. One can obtain
a more refined picture as to how the hardness of cardinality single-sink unsplittable/confluent flow varies with
 the demand sizes, or more precisely how it varies on the bottleneck
 value $\frac{d_{max}}{u_{min}}$.\footnote{This seems likely connected to  a footnote in \cite{azar2001strongly}
 that a lower bound of the
 form $O(m^{\frac12-\epsilon}\cdot \sqrt{\log (\frac{d_{max}}{u_{min}}}))$ exists for maximum unsplittable flow in
 directed graphs. Its proof was omitted however.}
Specifically, combining the approaches used in Theorems \ref{thm:extended} and \ref{thm:hard} gives:
\begin{thm}
\label{thm:harder} Consider any fixed $\epsilon > 0$ and  $d_{max}/u_{min} > 1$.
It is NP-hard to approximate cardinality single-sink unsplittable/confluent flow
to within a factor of $O(m^{\frac12-\epsilon}\cdot \sqrt{\log (\frac{d_{max}}{u_{min}})})$
in undirected or directed graphs. For unsplittable flow, this remains true for planar graphs.
\end{thm}

Once again we see the message that there is a sharp cutoff for  $d_{max}/u_{min} > 1$ even in the
large-demand regime. This is because
 if the bottleneck value is at most $1$, then the
no-bottleneck assumption holds and, consequently, the single-sink unsplittable flow problem
admits a constant-factor approximation (not $\sqrt{m}$ hardness).
We mention that a similar hardness bound cannot hold   for the maximum throughput objective, since one can
always reduce to the case where $d_{max}/u_{min}$ is small with a polylogarithmic loss, and hence the lower
bound becomes at worst $O(m^{\frac12-\epsilon}\cdot \log m)$. We feel the preceding hardness bound is all the more
interesting since known greedy techniques
yield almost-matching  upper bounds, even for general multiflows.
\begin{thm}
\label{thm:upper}
There is an $O(\sqrt{m}\log (\frac{d_{max}}{u_{min}}))$ approximation algorithm
for cardinality unsplittable flow and an $O(\sqrt{m}\log n)$ approximation algorithm
for throughput unsplittable flow, in both directed and undirected graphs.
\end{thm}

We next present one hardness result for confluent flows assuming the no-bottleneck-assumption.
Again, recall that for the maximum single-sink unsplittable flow problem
there is a constant factor approximation algorithm given the no-bottleneck-assumption.
We prove this is not the case for the single-sink confluent flow problem by providing a super-constant lower bound.
Its proof is more complicated but builds on the techniques used for our previous results.
\begin{thm}\label{thm:hardnba}
Given the no-bottleneck assumption, the single-sink confluent flow problem
cannot be approximated to within a factor $O(\log^{1-\epsilon}n)$, for any $\epsilon > 0$, unless $P=NP$.
This holds for both the maximum cardinality and maximum throughput objectives
in undirected and directed graphs.
\end{thm}

Finally, we include a hardness result for the congestion minimization problem for confluent flows.
That is, the problem of finding the minimum value $\alpha \geq 1$ such that all demands can be routed confluently if all node capacities are multiplied by $\alpha$.
This problem has two plausible variants.

An  {\em $\alpha$-congested}  routing is an unsplittable flow for the demands where the total load on any node is at most $\alpha$ times its associated capacity.
A {\em strong congestion} algorithm is one where  the resulting flow must route on a tree $T$ such that for  any demand $v$ the nodes on its path in $T$ must have capacity at least $d(v)$. A  {\em weak congestion} algorithm does not require this extra constraint on the tree capacities.
 Both variants are of possible interest. If the motive for congestion is to route all demands in some limited number $\alpha$ of rounds of admission, then each round should be feasible on $T$ - hence strong congestion is necessary.  On the other hand, if the objective is to simply augment network capacity so that all demands can be routed, weak congestion is the right notion.
In Section~\ref{sec:congestion} we show that it is hard to approximate strong congestion to within polynomial factors.
 \begin{thm}
 \label{thm:strongcongestion}
 It is NP-hard to approximate the minimum (strong) congestion problem for single-sink confluent flow instances
 (with polynomial-size demands) to factors of at most $m^{.5-\epsilon}$ for any $\epsilon>0$.
 \end{thm}

\subsection{Overview of Paper}
At the heart of our reductions are gadgets based upon the {\em capacitated} $2$-disjoint paths problem.
We discuss this problem in Section \ref{sec:two-disjoint-paths}.
In Section~\ref{sec:lower}, we prove the $\sqrt{m}$ hardness of maximum single-sink unsplittable/confluent
flow in the small demand regime (Theorem \ref{thm:hard}); we give a similar hardness for single-sink priority flow
(Corollary \ref{cor:priority}).
Using a similar basic construction, we prove, in Section~\ref{sec:confwithNBA}, the logarithmic hardness
of maximum single-sink confluent flow
even given the no-bottleneck assumption (Theorem \ref{thm:hardnba}).
In Section~\ref{sec:stronger}, we give lower bounds
on the approximability of the cardinality objective for general demand
regimes (Theorems \ref{thm:extended} and \ref{thm:harder}).
Finally, in Section~\ref{sec:upper}, we present an almost matching upper bound for unsplittable flow (Theorem \ref{thm:upper}).
and priority flow.

\section{The Two-Disjoint Paths Problem}\label{sec:two-disjoint-paths}

Our hardness reductions require gadgets based upon the {\em capacitated} $2$-disjoint paths problem.
Before describing this problem, recall the classical $2$-disjoint paths problem:\\

\noindent {\tt 2-Disjoint Paths (Uncapacitated):} Given a graph $G$ and node pairs $\{x_1, y_1\}$
and $\{x_2, y_2\}$. Does $G$ contain paths $P_1$ from $x_1$ to $y_1$
and $P_2$ from $x_2$ to $y_2$ such that $P_1$ and $P_2$ are disjoint?\\

Observe that this formulation incorporates four distinct problems
because the graph $G$ may be directed or undirected and the desired paths
may be edge-disjoint or node-disjoint.
In undirected graphs the $2$-disjoint paths problem, for both edge-disjoint and node disjoint paths,
can be solved in polynomial time
-- see  Robertson and Seymour~\cite{RS95}.
In directed graphs, perhaps surprisingly, the problem is NP-hard. This is the
case for both edge-disjoint and node disjoint paths, as shown
by  Fortune, Hopcroft and Wyllie~\cite{FHW80}.

In general, the unsplittable and confluent flow problems concern capacitated graphs.
Therefore, our focus is on the capacitated version of the $2$-disjoint paths problem. \\

\noindent {\tt 2-Disjoint Paths (Capacitated):} Let $G$ be a graph whose edges have
capacity either $\a$ or $\b$, where $\b \ge \a$.
Given node pairs $\{x_1, y_1\}$
and $\{x_2, y_2\}$, does $G$ contain paths $P_1$ from $x_1$ to $y_1$
and $P_2$ from $x_2$ to $y_2$ such that:\\
(i) $P_1$ and $P_2$ are disjoint.\\
(ii) $P_2$ may only use edges of capacity $\b$. ($P_1$
may use both capacity $\a$ and capacity $\b$ edges.) \\

For directed graphs, the result of Fortune et al.~\cite{FHW80} immediately implies that the capacitated version is
hard -- simply assume every edge has capacity $\beta$. In undirected graphs, the case of node-disjoint paths
was proven to be hard by Guruswami et al.~\cite{guruswami2003near}. The case of edge-disjoint paths was recently
proven to be hard by Naves, Sonnerat and Vetta~\cite{naves2010maximum}, even in
planar graphs where terminals lie on the outside face (in an interleaved order, which will be important for us). These results are summarised in Table \ref{table:hardness}.
\begin{table}[h]
      \centering
          \begin{tabular}{|c|cc|}
          \hline
& Directed & Undirected \\
\hline
Node-Disjoint  &    NP-hard \cite{FHW80} & NP-hard \cite{guruswami2003near} \\
Edge-Disjoint  & NP-hard \cite{FHW80} & NP-hard  \cite{naves2010maximum} \\
\hline
   \end{tabular}
\caption{Hardness of the Capacitated 2-Disjoint Paths Problem}\label{table:hardness}
\end{table}

Recall that the unsplittable flow problem has capacities on edges, whereas
the confluent flow problem has capacities on nodes.
Consequently, our hardness reductions for
unsplittable flows  require gadgets based upon
the hardness for edge-disjoint paths \cite{naves2010maximum};
for confluent flows we  require gadgets based upon
the hardness for node-disjoint paths \cite{guruswami2003near}.

\section{Polynomial Hardness of Single-Sink Unsplittable,\\ Confluent and Priority Flow}
\label{sec:lower}

In this section, we establish that the single-sink maximum unsplittable and confluent flow problems
are hard to approximate within polynomial factors for both the cardinality and throughput objectives.
We will then show how these hardness results extend to the single-sink maximum priority flow problem.
We begin with the small demand regime by proving Theorem~\ref{thm:hard}.
Its proof introduces some core ideas that are used in later sections
in the proofs of Theorems \ref{thm:harder} and \ref{thm:hardnba}.

\subsection{$\sqrt{n}$-Hardness in the Small Demand Regime}
\label{sec:hardness}


Our approach uses a grid routing structure much as in the hardness proofs of Guruswami et al.~\cite{guruswami2003near}.
Specifically:

(1) We introduce a graph $G_N$ that has the following properties.
There is a set of pairwise crossing paths that can route demands of total value,
$\sum_{i=1}^{N} (1+ i \delta) =N + \delta \frac12 N(N+1)$.  On the other hand, any collection of pairwise non-crossing
paths can route at most $d_{max} = 1+N\delta$ units of the total demand. For a given  $\Delta \in (0,1)$ we
choose $\delta$ to be small enough so that $d_{max} \leq 1+\Delta < 2$.

(2) We then build a new network $\mathcal{G}$ by replacing each node of $G_N$ by
an instance of the capacitated $2$-disjoint paths problem. This routing problem is chosen
because it induces the following properties.  If it is a YES-instance, then a maximum unsplittable (or confluent) flow on
$\mathcal{G}$ corresponds to routing demands in $G_N$ using pairwise-crossing paths. In contrast, if it is a
NO-instance, then a maximum unsplittable or confluent flow on $\mathcal{G}$  corresponds to routing
demands in $G_N$ using pairwise non-crossing paths.

Since $G_N$ contains $n=O(N^2)$ nodes, it follows that an approximation algorithm with guarantee better than
$\Theta(\sqrt{n})$  allows us to distinguish between YES- and NO-instances of our
routing problem, giving an inapproximability lower bound of $\Omega(\sqrt n)$.
Furthermore, at all stages we  show how this reduction can be applied using
only undirected graphs.  This will prove Theorem \ref{thm:hard}.

\subsubsection{A Half-Grid Graph $G_N$}

\begin{figure}[th]
\begin{center}
\includegraphics[height=12cm]{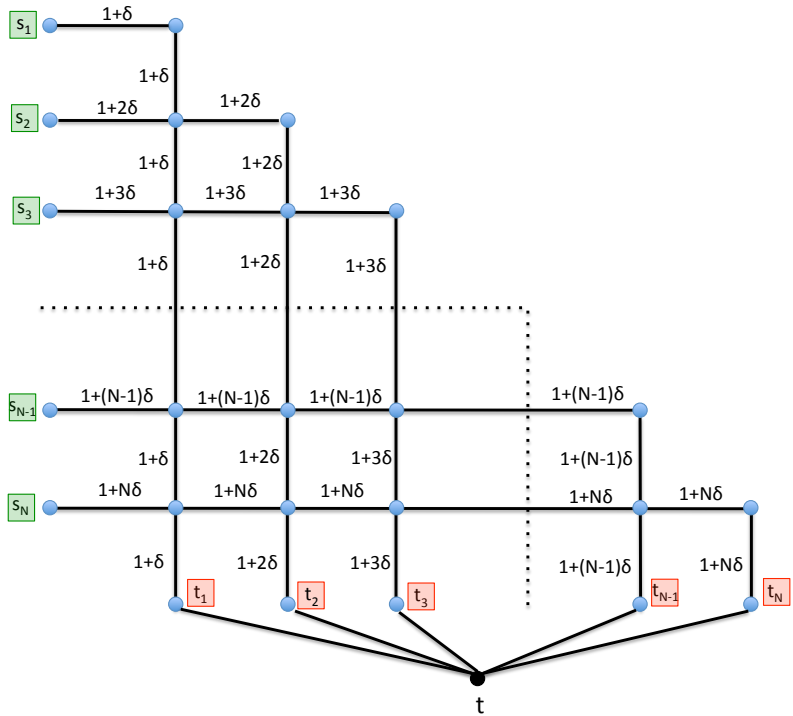}
\caption{\label{fig.grid} A Half-Grid $G_N$.}
\end{center}
\end{figure}

Let's begin by defining the graph $G_N$.  There are $N$ rows (numbered from
top to bottom) and $N$ columns (numbered from left to right). We call the leftmost node in the $i^{th}$ row $s_i$, and
the bottom node in the $j^{th}$ column $t_j$. There is a demand of size $\capp{i} := 1+i\delta$  located at $s_i$.
Recall, that $\delta$ is chosen so that all demands and capacities
lie within a tight range $[1,1+\Delta]$ for fixed $\Delta$ small. All the edges
in the $i^{th}$ row and all the edges in the $i^{th}$ column have capacity
$\capp{i}$.  The $i^{th}$ row extends as far as the $i^{th}$ column and vice
versa; thus, we obtain a ``half-grid" that is a weighted version of the
network considered by Guruswami et al. \cite{guruswami2003near}.  Finally we add a sink $t$.
There is an edge of capacity $\capp{j}$ between $t_j$ to $t$. The
complete construction is shown in Figure~\ref{fig.grid}.

For the unsplittable flow problem we have edge capacities. We explain later how the node capacities
are incorporated for the confluent flow problem. We also argue about the undirected and directed reductions together.
For directed instances we always enforce  edge directions to be downwards
and to the right.

Note that there is a unique $s_i-t$ path $P^*_i$ consisting only of edges of
capacity $\capp{i}$, that is, the hooked path that goes from $s_i$ along the
$i^{th}$ row and then down the $i^{th}$ column to $t$. We call this the {\em canonical} path
for demand $i$.

\begin{claim}\label{cl:cross}
Take two feasible paths $Q_i$ and $Q_j$ for demands $i$ and $j$.
If $i<j$, then the paths must cross on row $j$, between columns $i$ and $j-1$.
\end{claim}
\noindent{\bf Proof.}
Consider demand $i$ originating at $s_i$. This demand cannot
use any edge in columns $1$ to $i-1$ as it is too large.
Consequently, any feasible path $Q_i$ for demand $i$ must include all of row $i$.
Similarly, $Q_j$ must contain all of row $j$. Row $j$ cuts off $s_i$ from the sink $t$, so
$Q_i$ must meet $Q_j$ on row $j$.
Demand $i$ cannot use an edge in row $j$ as demand $j$ is already using up all the capacity along that row.
Thus $Q_i$ crosses $Q_j$ at the point they meet. As above, this meeting cannot occur
in columns $1$ to $i-1$. Thus the crossing point must occur on some column between $i$ and $j-1$
(by construction of the half-grid, column $j$ only goes as high as row $j$ so the crossing cannot be there).
\qed

By Claim \ref{cl:cross}, if we are forced to route using pairwise non-crossing paths, then only one demand
can route. Thus we can route at most a total of $\capp{N}=1+\delta N<2$ units of demand.

\subsubsection{The Instance $\mathcal{G}$}
We build a new instance $\mathcal{G}$ by replacing each degree $4$ node in $G_N$ with an instance
of the $2$~-~disjoint paths problem.
For the unsplittable flow problem in undirected graphs we use gadgets $H$ corresponding to
the capacitated edge-disjoint paths problem. Observe that a
node at the intersection of column $i$ and row $j$ (with $j > i$) in $G_N$ is incident to two edges of capacity
$\capp{i}$ and to two edges of weight $\capp{j}$.  We
construct  $\mathcal{G}$ by replacing each such node of degree four with
the routing graph $H$. We do this in such a way that the capacity $\capp{i}$
edges of $G_N$ are incident to $x_1$ and $y_1$, and the  $\capp{j}$
edges are incident to $x_2$ and $y_2$.
We also let $\a=\capp{i}$ and $\b=\capp{j}$.

For the confluent flow problem in undirected graphs we now have node capacities. Hence we use gadgets $H$ corresponding to
the node-capacitated 2-paths problem discussed above.
Again $x_1$ and $y_1$ are given
capacity $\capp{i}$ whilst $x_2$ and $y_2$ have capacity $\capp{j}$.

For directed graphs, the mechanism is simpler as the gadgets may now come from the
uncapacitated disjoint paths problem. Thus the hardness comes from the directedness and not from the capacities.
Specifically, we may set the edge capacities to be $C=\max\{\capp{i},\capp{j}\}$. Moreover, for unsplittable flow we may
perform the standard operation of splitting each node in $H$ into two, with
 the new induced arc having capacity of $C$. It follows that if there are two flow paths through $H$, each carrying at
 least $\capp{i} \geq \capp{j}$ flow, then they must be from $x_1$ to $y_1$ and $x_2$ to $y_2$.
 These provide a solution to the node-disjoint directed paths problem in $H$.

The hardness result will follow once we see how this construction relates to
crossing and non-crossing collections of paths.

\begin{lemma}\label{lem:yes}
If $H$ is a YES-instance, then the maximum unsplittable/confluent flow in
$\mathcal{G}$ has value at least $N$.  For a NO-instance the
maximum unsplittable/confluent flow has value at most $1+\Delta < 2$.
\end{lemma}
\noindent{\bf Proof.}  If $H$ is a YES-instance, then we can use its paths to produce
paths in $\mathcal{G}$, whose images in $G_N$,   are free to cross at any node.
Hence we can produce paths in $\mathcal{G}$ whose images are
the canonical paths $P^*_i, \, 1 \le i \le N$ in $G_N$.
This results in   a flow of value greater than
$N$.  Note that  in the confluent case, these paths yield
 a confluent flow as they only  meet at the root $t$.

Now suppose $H$ is a NO-instance.
Take any flow and consider two paths $\hat{Q}_i$ and $\hat{Q}_j$ in $\mathcal {G}$ for demands $i$ and $j$, where $i < j$.
These paths also induce two feasible paths $Q_i$ and $Q_j$ in the half-grid $G_N$.
By Claim \ref{cl:cross}, these paths cross on row $j$ of the half-grid (between columns $i$ and $j-1$).
In the directed case (for unsplittable or confluent flow) if they cross at a grid-node $v$, then the paths they induce in the copy of $H$ at $v$
must be node-disjoint. This is not possible in the directed case since
such paths do not exist for $(x_1,y_1)$ and $(x_2,y_2)$.

In the undirected confluent case, we must also have node-disjoint paths through this copy of $H$.
 As we are in row $j$ and a column between column $i$ and $j-1$,
we have $\b=\capp{j}$ and $\capp{i} \le \a \le \capp{j-1}$. Thus, demand $j$ can only use the $\b$-edges of $H$.
This contradicts the fact that $H$ is a NO-instance. For the undirected case of unsplittable flow the two
paths through $H$ need to be edge-disjoint, but now we obtain a contradiction as
our gadget was derived from the
capacitated edge-disjoint paths problem.

It follows that no such pair $\hat{Q}_i$ and $\hat{Q}_j$ can exist and, therefore, the confluent/unsplittable flow routes at
most one demand and, hence, routes a total demand of at most $1+\Delta$.
\qed

We then obtain our hardness result. \\

{{\noindent\bf Theorem~\ref{thm:hard}.} \itshape
Neither cardinality nor throughput can be approximated to
within a factor of $O(m^{\frac12-\epsilon})$, for any
$\epsilon > 0$, in the single-sink unsplittable and confluent flow problems.
This holds for  undirected and directed graphs even
when instances are restricted to have bottleneck value
$\frac{d_{max}}{u_{min}}=1 + \Delta$ where $\Delta  > 0$ is arbitrarily small.\\
}

\noindent{\bf Proof.} It follows that if we could
approximate the maximum (unsplittable) confluent flow problem in $\mathcal{G}$ to a
factor better than $N$, we could determine whether the optimal solution is
at least $N$ or at most $1+\Delta$. This in turn would allow us to determine whether $H$ is a YES-
or a NO-instance.

Note that $\mathcal{G}$ has $n=\Theta(pN^2)$ edges, where $p=|V(H)|$. If we
take $N = \Theta(p^{\frac12(\frac{1}{\epsilon}-1)})$, where $\epsilon>0$ is an (arbitrarily)
small constant, then $n=p^{\frac{1}{\epsilon}}$ and so
$N = \Theta(n^{\frac12 (1-\epsilon)})$.   A precise lower bound of $n^{.5 - \epsilon'}$ is obtained for
$\epsilon' > \epsilon$ sufficiently small, when $n$ is sufficiently large.
\qed


\subsubsection{Priority Flows and Congestion}
\label{sec:congestion}

We now show the hardness of priority flows. To do this,
we use the same half-grid construction, except we must replace the
capacities by priorities. This is achieved in a straight-forward manner:
priorities are defined by the magnitude of the original demands/capacities.
The larger the demand or capacity in the original instance, the higher its priority in the
new instance. (Given the priority ordering we may then assume all demands and
capacities are set to $1$.)
In this setting, observe that Claim~\ref{cl:cross} also applies for priority flows.
\begin{claim}
Consider  two feasible paths $Q_i$ and $Q_j$ for demands $i$ and $j$ in the priority flow problem.
If $i<j$, then the paths must cross on row $j$, between columns $i$ and $j-1$.
\end{claim}
\noindent{\bf Proof.}
Consider demand $i$ originating at $s_i$. This demand cannot
use any edge in columns $1$ to $i-1$ as they do not have high enough priority.
Consequently, any feasible path $Q_i$ for demand $i$ must include all unit capacity edges  of row $i$.
Similarly, $Q_j$ must contain all of row $j$. Row $j$ cuts off $s_i$ from the sink $t$, so
$Q_i$ must meet $Q_j$ on row $j$.
Demand $i$ cannot use an edge in row $j$ as demand $j$ is already using up all the capacity along that row.
Thus $Q_i$ crosses $Q_j$ at the point they meet. As above, this meeting cannot occur
in columns $1$ to $i-1$. Thus the crossing point must occur on some column between $i$ and $j-1$.
\qed

Repeating our previous arguments, we obtain the following hardness result for priority flows.
(Again, it applies to both throughput and cardinality objectives as they coincide for priority flows.)

{{\noindent\bf Corollary~\ref{cor:priority}.} \itshape
The maximum single-sink priority flow problem
cannot be approximated to within a factor of $m^{\frac12-\epsilon}$, for any $\epsilon > 0$, in planar
directed or undirected graphs. \qed
}

\vspace*{.2cm}
We close the section by establishing   Theorem~\ref{thm:strongcongestion}. Consider  grid instance  built from
 a YES instance of the 2 disjoint path problem. As before we may find a routing of all demands with  congestion at most $1$.
Otherwise, suppose that the grid graph is built from  a NO instance and consider a tree $T$ returned by a strong congestion algorithm.
As it is a strong algorithm, the demand in row $i$ must follow its canonical path horizontally to the right as far as it can.
As it is a confluent flow, all demands from rows $>i$ must accumulate at this rightmost node in row $i$.
Inductively this implies that the total load at the rightmost node in row 1 has load $>N$.
As before, for any $\epsilon > 0$ we may choose $N$ sufficiently large so that $N \geq n^{.5-\epsilon}$. Hence
we have a YES instance of 2 disjoint paths if and only if the output from a $n^{.5-\epsilon}$-approximate strong congestion algorithm
returns a solution with congestion $\leq N$.

\section{Logarithmic Hardness of Single-Sink Confluent Flow with the No-Bottleneck Assumption}
\label{sec:confwithNBA}

We now prove the logarithmic hardness of the confluent flow problem given
the no-bottleneck assumption.
A similar two-step plan is used as for Theorem~\ref{thm:hard} but the analysis is
more involved.

(1) We introduce a planar graph $G_N$ which has the same structure as our previous half-grid, except
that its edge weights are changed. As before we  have demands associated with the $s_i$'s, but we assume these demands
are tiny -- this  guarantees that the no-bottleneck assumption holds.
We thus refer to the demands located at an $s_i$ as the {\em  packets} from $s_i$.
We define $G_N$ to ensure that there is a
collection of pairwise crossing ``trees'' (to be defined) that can route
packets of total value equal to the harmonic number $H_N\approx \log N$.
On the other hand, any collection of pairwise non-crossing trees can route at
most one unit of packet demand.

(2) We then build a new network $\mathcal{G}$ by replacing each node of $G_N$ by
an instance of the $2$-disjoint paths problem. Again, this routing problem is chosen
because it induces the following properties. If it is a YES-instance, then we can find a routing that
corresponds to pairwise crossing trees. Hence we are able to route $H_N$ demand.
In contrast, if it is a NO-instance, then a maximum confluent flow on $\mathcal{G}$ is forced to route using a non-crossing
structure and this forces the total flow to be at most $1$.

It follows that an approximation algorithm with guarantee better than
logarithmic would allow us to distinguish between YES- and NO-instances of our
routing problem, giving a lower bound of $\Omega(\log N)$. We will see that
this bound is equal to $\Theta(\log^{1-\epsilon} n)$.


\subsection{{An Updated Half-Grid Graph.}}\ \\
Again we take the graph $G_N$ with $N$ rows (now numbered from
 bottom to top) and $N$ columns (now numbered from right to left).  All the edges
in the $i^{th}$ row and all the edges in the $i^{th}$ column have capacity
$\frac{1}{i}$.  The $i^{th}$ row extends as far as the $i^{th}$ column and vice
versa; thus, we obtain a half-grid similar to our earlier construction but with updated  weights.
 Then we add a sink $t$.
There is an edge of capacity $\frac{1}{i}$ to $t$ from the bottom node (called $t_i$) in column $i$.
Finally, at the leftmost node (called $s_i$) in row $i$ there is a collection of packets (``sub-demands'')
whose total weight is $\frac{1}{i}$. These packets are very small. In particular,
they are much smaller than $\frac{1}{n}$, so they satisfy the no-bottleneck assumption.
The complete construction is shown in Figure \ref{fig.grid2}.
In the directed setting, edges are oriented to the right and downwards.

\begin{figure}[h]
\begin{center}
\includegraphics[height=10cm]{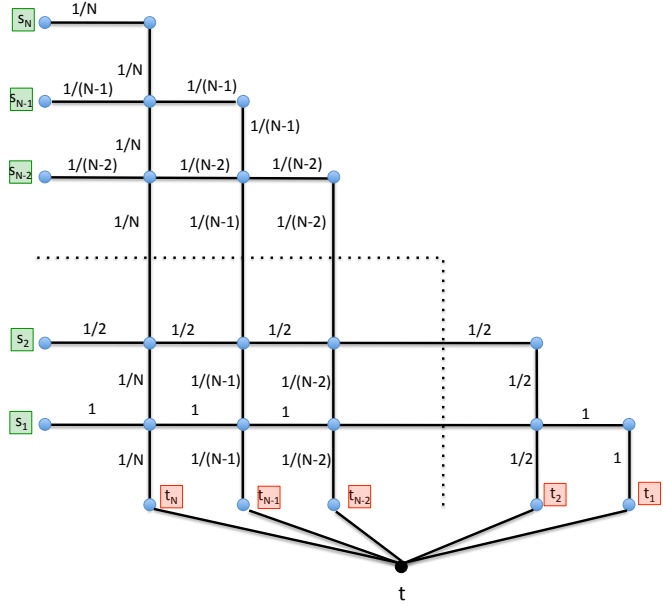}
\caption{\label{fig.grid2} An Updated \nba \ Half-Grid $G_N$.}
\end{center}
\end{figure}

Again, there is a unique $s\mbox{-}t$ path $P^*_i$ consisting only of edges of
weight $\frac{1}{i}$, that is, the hooked path that goes from $s_i$ along the
$i$th row and then down the $i^{th}$ column to $t$.
Moreover, for $i \neq j$, the path $P^*_i$ intersects $P^*_j$ precisely once.
If we route packets along the paths
$\mathcal{P}^*=\{P^*_1,P^*_2,\dots,P^*_N\}$, then we obtain a  flow of total value $H_N
=1+\frac12+\ldots \frac{1}{N}$.  Since every edge incident to $t$ is used in
$\mathcal{P}^*$ with its maximum capacity, this solution is a maximum single-sink flow.
Clearly, each $P^*_i$ is a tree, so this routing corresponds to our notion of routing on ``crossing trees''.

We then build $\mathcal{G}$ as before by replacing the degree four nodes in the grid by our
disjoint-paths gadgets. Our first objective is to analyze the maximum flow possible in the case
where our derived instance $\mathcal{G}$ is made from NO-instances. Consider a confluent flow in $\mathcal{G}$.
If we contract the pseudo-nodes, this results in some {\em leaf-to-root} paths in the subgraph $G_N$.
We define $T_i$ as the union of all such leaf-to-root paths terminating at $t_i$.
If we have a NO-instance, then  the resulting collection  $\mathcal{T}=\{T_1,T_2,\dots,T_N\}$ forms non-crossing subgraphs.
That is, if $i \neq j$, then there do not exist leaf-to-root paths $P_i \in T_i$ and $P_j \in T_j$
which cross in the standard embedding of $G_N$.
Since we started with a confluent flow in $\mathcal{G}$, the flow paths within each $T_i$ are {\em edge-confluent}.
That is, when two flow paths share an {\bf edge}, they must thereafter follow the same path to $t_i$.
Note that they could meet and diverge at a node if they
use different incoming and outgoing edges.
In the following,
we identify the subgraph $T_i$ with its edge-confluent flow.

The {\em capacity} of a $T_i$ is the maximum flow from its leaves
to $t_i$. The capacity of a collection $\mathcal{T}$ is then the sum of these capacities.
We first prove that the maximum value of a flow (i.e., capacity) is significantly reduced if we require a
non-crossing collection of edge-confluent flows.
One should note that as our demands are tiny, we may essentially route any partial
amount $x \leq \frac{1}{i}$ from a node $s_i$;
we cannot argue as though we route the whole $\frac{1}{i}$. On the other hand, any
packets from $s_i$ must route on the same path, and in particular $s_i$ lies in a unique  $T_j$ (or none at all).
Another subtlety in the proof is to handle the fact that we cannot apriori assume
that there is at most one leaf $s_j$ in a $T_i$.
 Hence such a flow does not just correspond to a maximum uncrossed unsplittable flow. In fact, because the packets are tiny,
 it is easy to verify that all the packets may  be routed unsplittably (not confluently) even if they are required to use non-crossing paths.

\begin{lemma}\label{lemma.maxflow}
The maximum capacity of a non-crossing edge-confluent flow  in $G_N$ is at most $2$.
\end{lemma}
\noindent{\bf Proof.}
Let $t_{i_1},t_{i_2},...,t_{i_k}$ be the roots of the subgraphs $T_i$ which support
the edge-confluent flow, where wlog $i_1 > i_2 > \cdots >i_k$. We argue inductively about the
topology of where these supports live in $G_N$.
For $i \leq j$ we define a subgrid $G(i,j)$ of $G_N$ induced by columns and rows whose
indices lie in the range $[i,j]$.
For instance, the rightmost column of $G(i,j)$ has capacity $\frac{1}{i}$ and  the
leftmost column $\frac{1}{j}$; similarly, the lowest row of $G(i,j)$ has capacity $\frac{1}{i}$ and  the
highest row $\frac{1}{j}$.

Obviously all the $T_i$'s route in $G(1,N)=G(r_1,\ell_1)$ where we define $r_1=1,\ell_1=N$.
Consider the topologically highest path $P_{i_1}$ in $T_{i_1}$,
 and let $r'_1$ be the highest row number where this path intersects column $n_1=t_{i_1}$.
We define $r_2 = r_1'+1$ and $\ell_2 = n_1-1$ and consider the subgrid $G(r_2,\ell_2)$.
Observe that in the undirected case it is possible that $P_{i_1}$ routes through the
subgrid $G(r_2,\ell_2)$; see Figure \ref{fig:upperT}(b). In the directed case this cannot happen;
see Figure \ref{fig:upperT}(a).

\begin{figure}[h]%
    \centering
    \subfloat[]{\includegraphics[width=8cm]{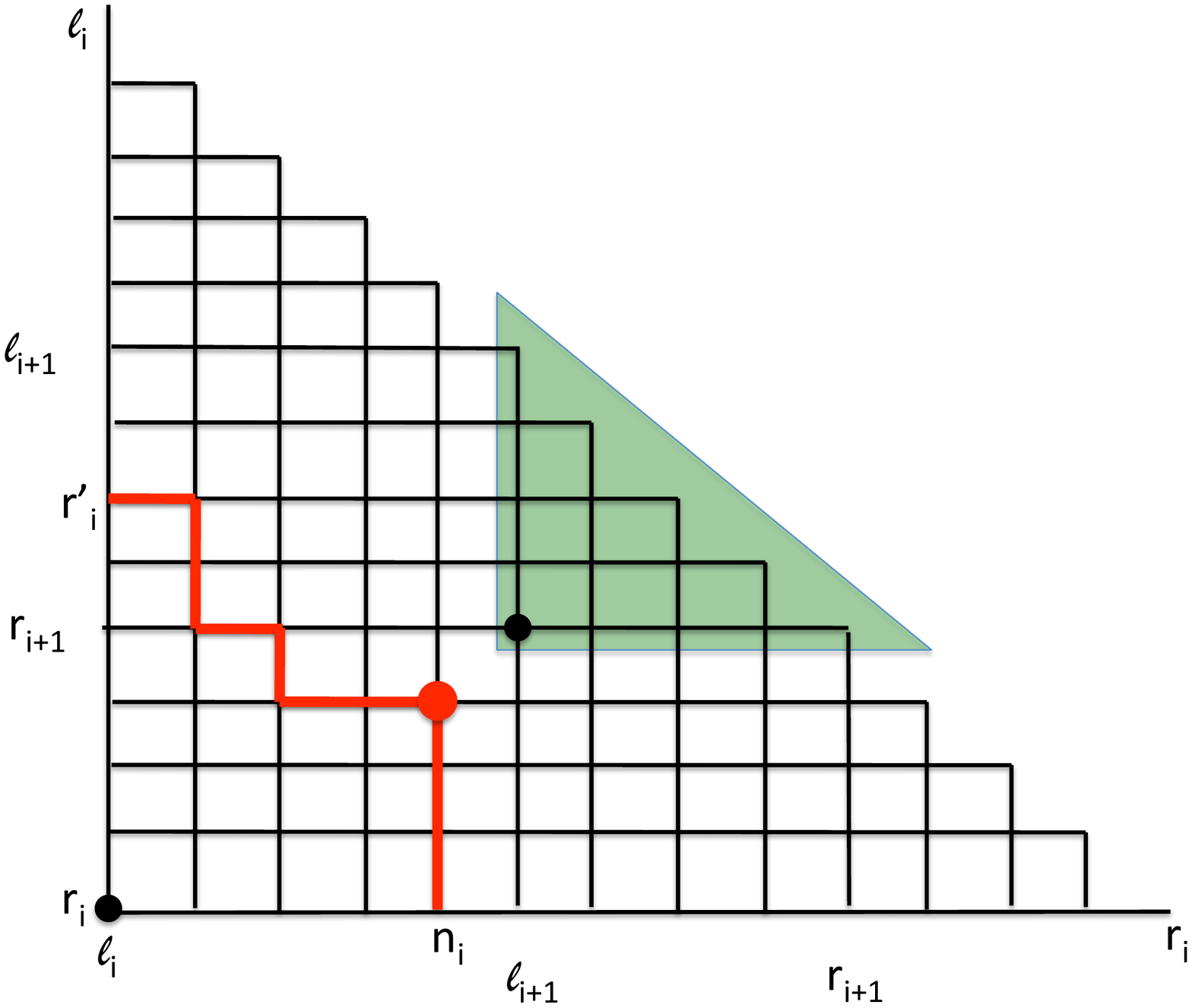} }%
    \qquad
    \subfloat[]{\includegraphics[width=8cm]{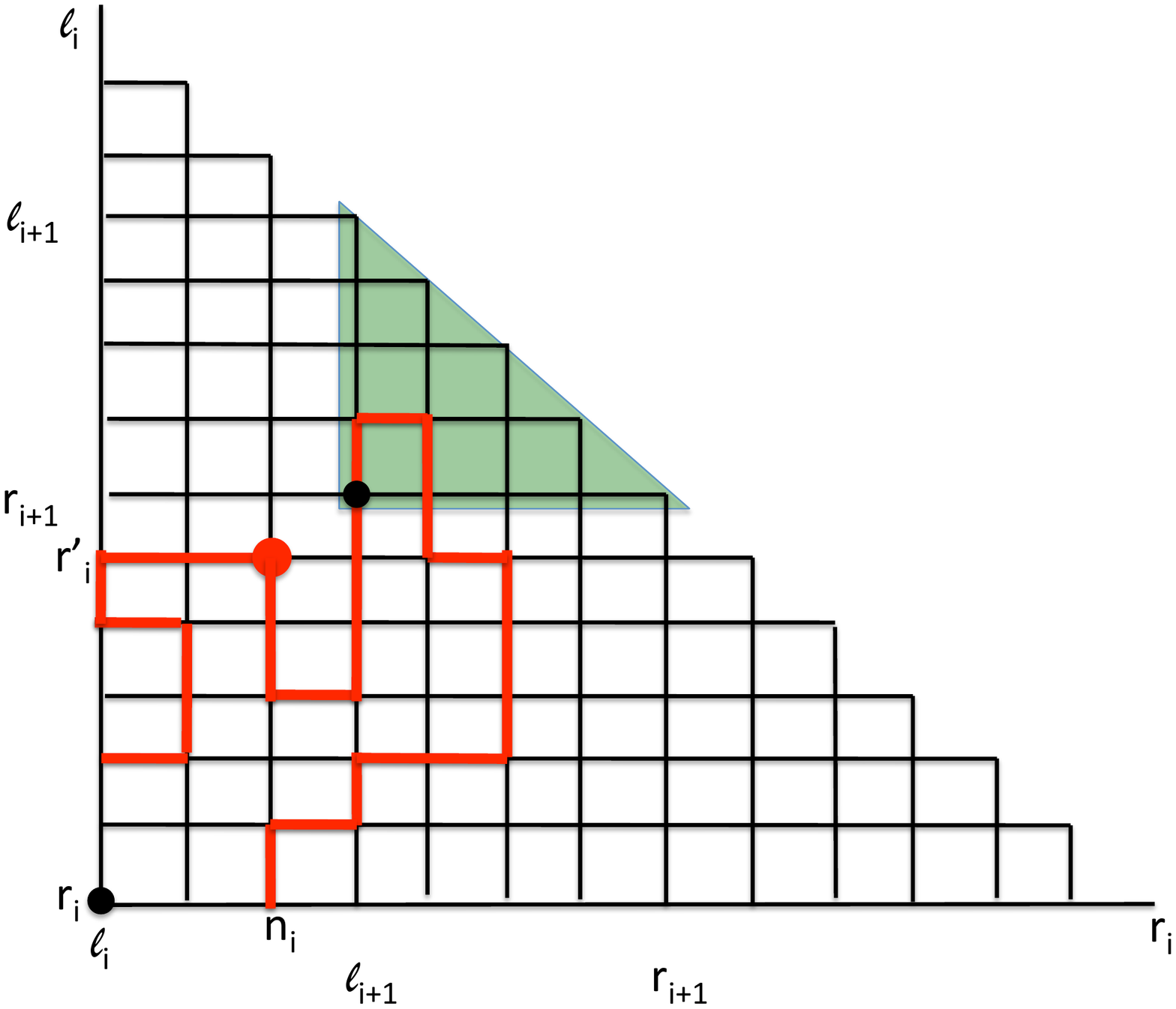}}%
    \caption{}%
    \label{fig:upperT}%
\end{figure}

In addition, it is possible that $T_{i_2}$ completely avoids routing through the subgrid $G(r_2,\ell_2)$.
But, for this to happen, it must have a cut-edge (containing all its flow) in column $t_{i_1}$;
consequently, its total flow is then at most $\frac{1}{i_1}$. It is also possible that it has some flow which
avoids $G(r_1,\ell_1)$ and some that does not.
Since $T_{i_1}$ also has maximum flow at most $\frac{1}{i_1}$, it follows that in
every case the total flow is at most $\frac{2}{i_1}$ plus the
maximum size of a confluent flow in the subproblem $G(r_2,\ell_2)$. Note that in this subproblem,
its ``first'' rooted subgraph may be at $t_{i_2}$ or $t_{i_3}$ depending on which of the two cases
above occurred for $T_{i_2}$.

If we iterate this process, at a general step $i$ we have a edge-confluent flow in the subgrid $G(r_i,\ell_i)$
whose lower-left corner is in row $r_i$ and column $\ell_i$ (hence $r_i \leq \ell_i$). Note that these triangular grids
are in fact nested.
Let $T_{n_i}$ be its subgraph rooted at a $t_{n_i}$ with $n_i$ maximized (that is, furthest to the left on bottom row).
As before, the total flow in this sub-instance is at most $\frac{2}{n_i}$ plus
a maximum edge-confluent flow in some $G(r_{i+1},\ell_{i+1})$. Since each new sub-instance
has at least one less rooted flow,
this process ends after at most $k^*\le k$ steps.  Note that for $i < k^*$ we have
$\frac{2}{n_i} \leq \frac{2}{\ell_{i+1}}$ and for $i=k^*$ we have $\frac{2}{n_{k^*}} \leq \frac{1}{r_{k^*}}$.
The latter inequality follows since for each $i$ we have $r_i \leq n_i \leq \ell_i$.

Now by construction we have the grids are nested and so
$\ell_1 > \ell_2 > \ldots \ell_{k^*} \geq r_{k^*} > \ldots r_2 > r_1$(recall that columns are ordered increasingly from
right to left). Since $r_1=1$, we may inductively deduce that
$r_{i} \geq i$ for all $i$. Thus $\ell_i \ge k^*$ for all $i$.
The total flow in our instance is then at most
\begin{eqnarray*}
2\cdot \sum_{1 \leq i \leq k^*}  \frac{1}{n_i} &\leq& 2\cdot (\sum_{2 \leq i \leq k^*} \frac{1}{\ell_i} + \frac{1}{r_{k^*}}) \\
&\le& 2\cdot \sum_{1 \leq i \leq k^*} \frac{1}{k^*} \\
&=& 2
\end{eqnarray*}
The lemma follows.
\qed

We can now complete the proof of the approximation hardness. Observe that any
node of degree four in $G_N$ is incident to two edges of weight
$\frac{1}{i}$ and to two edges of weight $\frac{1}{j}$, for some $j < i$. Again, we
construct a graph $\mathcal{G}$ by replacing each node of degree four with
an instance $H$ of the $2$ node-disjoint paths problem,
where the weight $\frac{1}{i}$
edges of $G_N$ are incident to $x_1$ and $y_1$, and the weight $\frac{1}{j}$
edges are incident to $x_2$ and $y_2$.
In the undirected case we require capacitated node-disjoint paths
and set $\alpha=\frac{1}{i}$ and $\beta=\frac{1}{j}$.
More precisely, since we are dealing with node capacities in confluent flows,
we actually subdivide each edge of $H$ and the new node
inherits the edge's capacity.
The nodes $x_1$ and $y_1$ also have capacity $\frac{1}{i}$ whilst the nodes
$x_2$ and $y_2$ have capacity $\frac{1}{j}$ in order to simulate the edge capacities of $G_N$.

\begin{lemma}\label{lem:yes2}
If $H$ is a YES-instance, then the maximum single-sink confluent flow in
$\mathcal{G}$ has value $H_N$.  If $H$ is a NO-instance, then the maximum
confluent flow in $\mathcal{G}$ has value at most $2$.
\end{lemma}
\noindent{\bf Proof.}  It is clear that if $H$ is a YES-instance, then the two feasible paths
in $H$ can be used to allow paths in $G_N$ to cross at any node
without restrictions on their values. This means we obtain a confluent flow of value
$H_N$ by using the canonical paths $P^*_i, \, 1 \le i \le N$.

Now  suppose that $H$ is a NO-instance and
consider how a confluent
flow $\mathcal{T}=\{T_1,\dots, T_n\}$ routes packets through the gadgets.
As it is a NO-instance,
the image of the  trees (after contracting the $H$'s to single nodes) in $G_N$ yields a non-crossing edge-confluent flow.
The capacity of this collection in $G_N$ is at least that in $\mathcal{G}$.
By Lemma~\ref{lemma.maxflow}, their capacity is at most $2$, completing the proof.
\qed

\ \\
{{\noindent\bf Theorem~\ref{thm:hardnba}.} \itshape
Given the no-bottleneck assumption, the single-sink confluent flow problem
cannot be approximated to within a factor $O(\log^{1-\epsilon}n)$, for any $\epsilon > 0$, unless $P=NP$.
This holds for both the maximum cardinality and maximum throughput objectives
in undirected and directed graphs.\\
}

\noindent{\bf Proof.}  It follows that if we could
approximate the maximum confluent flow problem in $\mathcal{G}$ to a
factor better than $H_N/2$, we could determine whether the optimal solution is
$2$ or $H_N$. This in turn would allow us to determine whether $H$ is a YES-
or a NO-instance.

Note that $\mathcal{G}$ has $n=\Theta(pN^2)$ edges, where $p=|V(H)|$. If we
take $N = \Theta(p^{\frac12(\frac{1}{\epsilon}-1)})$, where $\epsilon>0$ is a
small constant, then $H_N=\Theta(\frac{1}{2}(\frac{1}{\epsilon}-1) \log p)$. For $p$ sufficiently large,
this is $\Omega((\log n)^{1-\epsilon}) = (\frac{1}{\epsilon} \log p )^{1-\epsilon}$.  This gives a
lower bound of $\Omega((\log n)^{1-\epsilon})$.  \qed

Similarly, if we are restricted to consider only flows that decompose into $k$
disjoint trees then it is not hard to see that:
\begin{thm}
Given the no-bottleneck assumption, there is a $\Omega(\log k)$ hardness of
approximation, unless $P=NP$, for the problem of finding a maximum confluent flow that
decomposes into at most $k$ disjoint trees. \qed
\end{thm}

\section{Stronger Lower Bounds for Cardinality Single-Sink Unsplittable Flow with Arbitrary Demands}\label{sec:stronger}

In the large demand regime even stronger lower bounds can be achieved for the cardinality objective.
To see this, we explain the technique of Azar and Regev \cite{azar2001strongly} (used to prove Theorem~\ref{thm:ar})
in Section \ref{sec:expo-demands} and show how to extend it to undirected graphs and to confluent flows.
Then in Section \ref{sec:refine}, we combine their construction with the half-grid graph
to obtain lower bounds in terms of the bottleneck value (Theorem~\ref{thm:harder}).

\subsection{$m^{1-\epsilon}$ Hardness in the Large-Demand Regime}\label{sec:expo-demands}

{{\noindent\bf Theorem~\ref{thm:extended}.} \itshape
If $P \neq NP$ then, for any $\epsilon > 0$, there is no $O(m^{1-\epsilon})$-approximation algorithm
for the cardinality objective of the single-sink unsplittable/confluent flow problem in undirected graphs.\\
}

\noindent{\bf Proof.}
We begin by describing the construction of Azar and Regev for directed graphs.
They embed instances of the uncapacitated $2$-disjoint paths problem into a directed path.
Formally, we start with a directed path  $z^1,z^2, \ldots ,z^{\ell}$ where
 $t=z^{\ell}$  forms our sink destination for all demands. In addition, for each $i < \ell$, there are two parallel edges
 from $z^{i-1}$ to $z^{i}$. One of these has capacity $2^{i}$ and the other has a smaller capacity of $2^{i}-1$.
 There is a demand $s_i$ from each $z^i$, $i < \ell$ to $z^{\ell}$ of size $2^{i+1}$.
 Note that this unsplittable flow instance is feasible as follows. For each demand $s_j$, we may follow the high capacity edge
 from $z^j$ to $z^{j+1}$ (using up all of its capacity) and then use low capacity edges on the
 path $z^{j+1},z^{j+2}, \ldots ,z^{\ell}$. Call these the {\em canonical paths} for the demands.
The total demand on the low capacity edge from $z^j$ is then $\sum_{i \leq j} 2^i =2^{j+1}-1$, as desired.

Now replace each node $z^j$, $1\le j < \ell$, by an instance $H^j$ of the uncapacitated directed $2$-disjoint paths problem.
 Each edge in $H^j$ is given capacity $2^{j+1}$. Furthermore:\\
 (i) The tail of the high capacity edge out of $z^j$ is identified with the node $y_2$. \\
 (ii)  The tail of the low capacity edge out of $z^j$ is identified with $y_1$.\\
 (iii) The heads of both edges into $z^j$ (if they exist) are identified with $x_1$.\\
 (iv) The node $x_2$ becomes the starting point of the demand $s_j$ from $z^j$. \\
This construction is shown in Figure \ref{fig:line}.

\begin{figure}[h]
\begin{center}
\includegraphics[width=16cm,height=11cm]{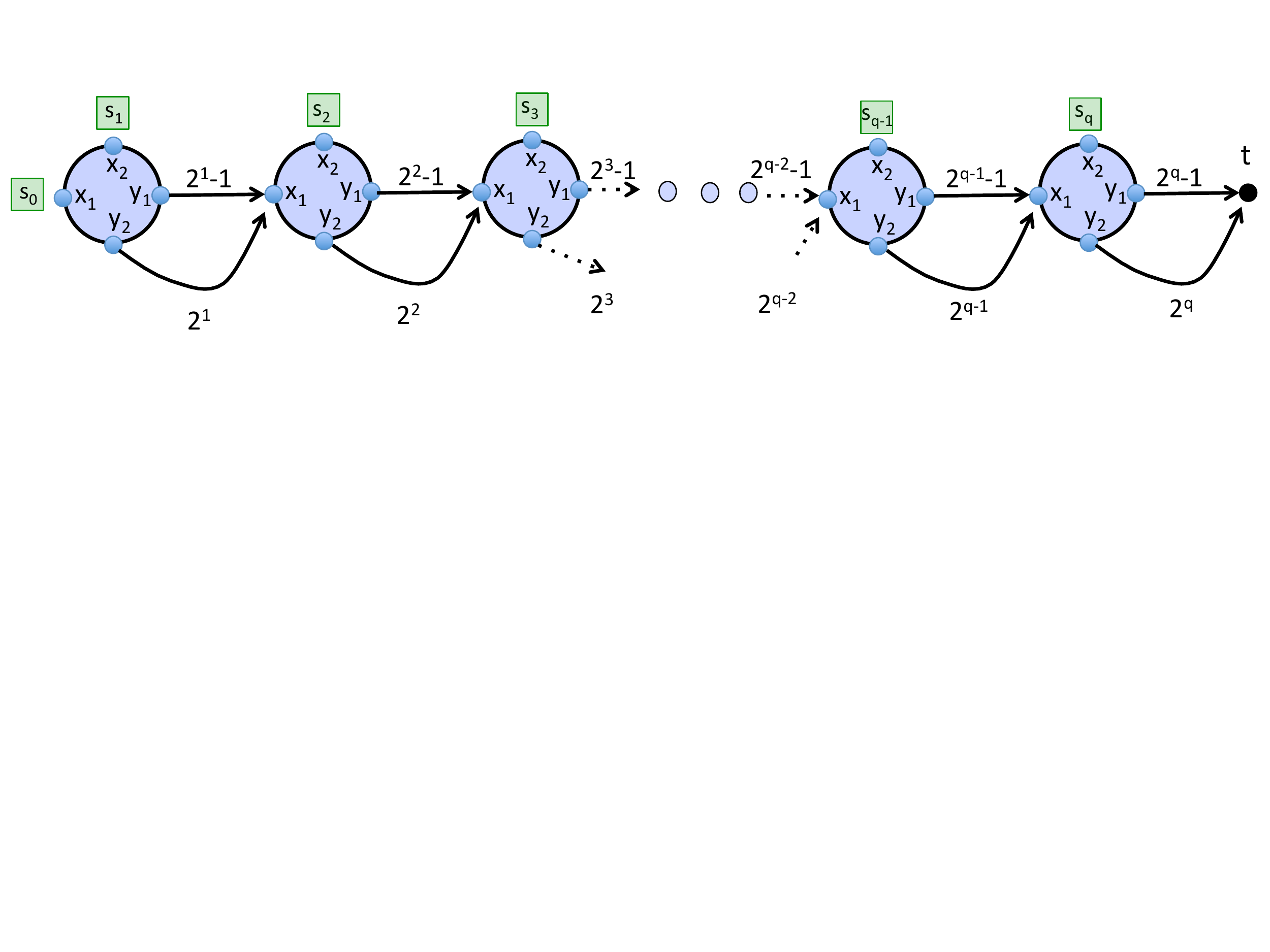}\\
\vspace{-6cm}
\caption{\label{fig:line} An Azar-Regev Path}
\end{center}
\end{figure}

Now if we have a YES-instance of the $2$-disjoint paths problem,
we may then simulate the canonical paths in the standard way. The demand in $H^j$ uses  the directed path
from $x_2$ to $y_2$ in $H^j$; it then follows the high capacity edge from $y_2$ to the $x_1$-node in the
 next instance $H^{j+1}$. All the total demand arriving from upstream $H^i$'s entered $H^{j}$ at its node $x_1$ and
 follows the directed path from $x_1$ to $y_1$. This total demand is at most $\sum_{i \leq j} 2^i$ and thus fits into the low
 capacity edge from $H^j$ into $H^{j+1}$.
Observe that this routing is also confluent in our modified instance because the paths in the $H^j$'s are node-disjoint.
 Hence, if we have a YES-instance
 of the $2$-disjoint paths problem, both the unsplittable and confluent flow problems have a solution routing all of the demands.

 Now suppose that we have a NO-instance, and consider a solution to the unsplittable (or confluent) flow problem.
 Take the highest value $i$ such that
 the demand from $H^i$ is routed. By construction, this demand must use a path $P_2$ from $x_2$ to $y_2$. But this
 saturates the high capacity edge from $y_2$. Hence  any demand from $H^j$, $j < i$ must pass from $y_1$ to $x_1$ while
 avoiding the edges of $P_2$. This is impossible, and so we route at most one demand.

 This gives a gap of $\ell$ for the cardinality objective. Azar-Regev then choose $\ell = |V(H)|^{\lceil \frac{1}{\epsilon} \rceil}$
to obtain a hardness of $\Omega(n^{1-\epsilon'})$.

Now consider undirected graphs. Here we use an
undirected instance of the capacitated $2$-disjoint paths problem.
We plug this instance into each $H^j$, and use the two capacity values of
$\b=2^{j+1}$ and $\a=2^{j+1}-1$.
A similar routing argument then gives the lower bound. \qed


We remark that it is easy to see why this approach does not succeed for the throughput objective.
The use of exponentially growing demands implies that high throughput is achieved simply
by routing the largest demand.

\subsection{Lower Bounds for Arbitrary Demands}\label{sec:refine}

By combining paths and half-grids we are able to refine the lower
bounds in terms of the bottleneck value (or demand spread).\\

{{\noindent\bf Theorem~\ref{thm:harder}.} \itshape
Consider any fixed $\epsilon > 0$ and  $d_{max}/u_{min} > 1$.
It is NP-hard to approximate cardinality single-sink unsplittable/confluent flow
to within a factor of $O(\sqrt{\log (\frac{d_{max}}{u_{min}})}\cdot m^{\frac12-\epsilon})$
in undirected or directed graphs. For unsplittable flow, this remains true for planar graphs.\\
}

\noindent{\bf Proof.}
We start with two parameters $p$ and $q$. We create $p$ copies of
the Azar-Regev path
 and attach them to a $p\times p$ half-grid, as shown in Figure~\ref{fig:gridline}.

\begin{figure}[h]
\begin{center}
\includegraphics[width=16cm,height=11cm]{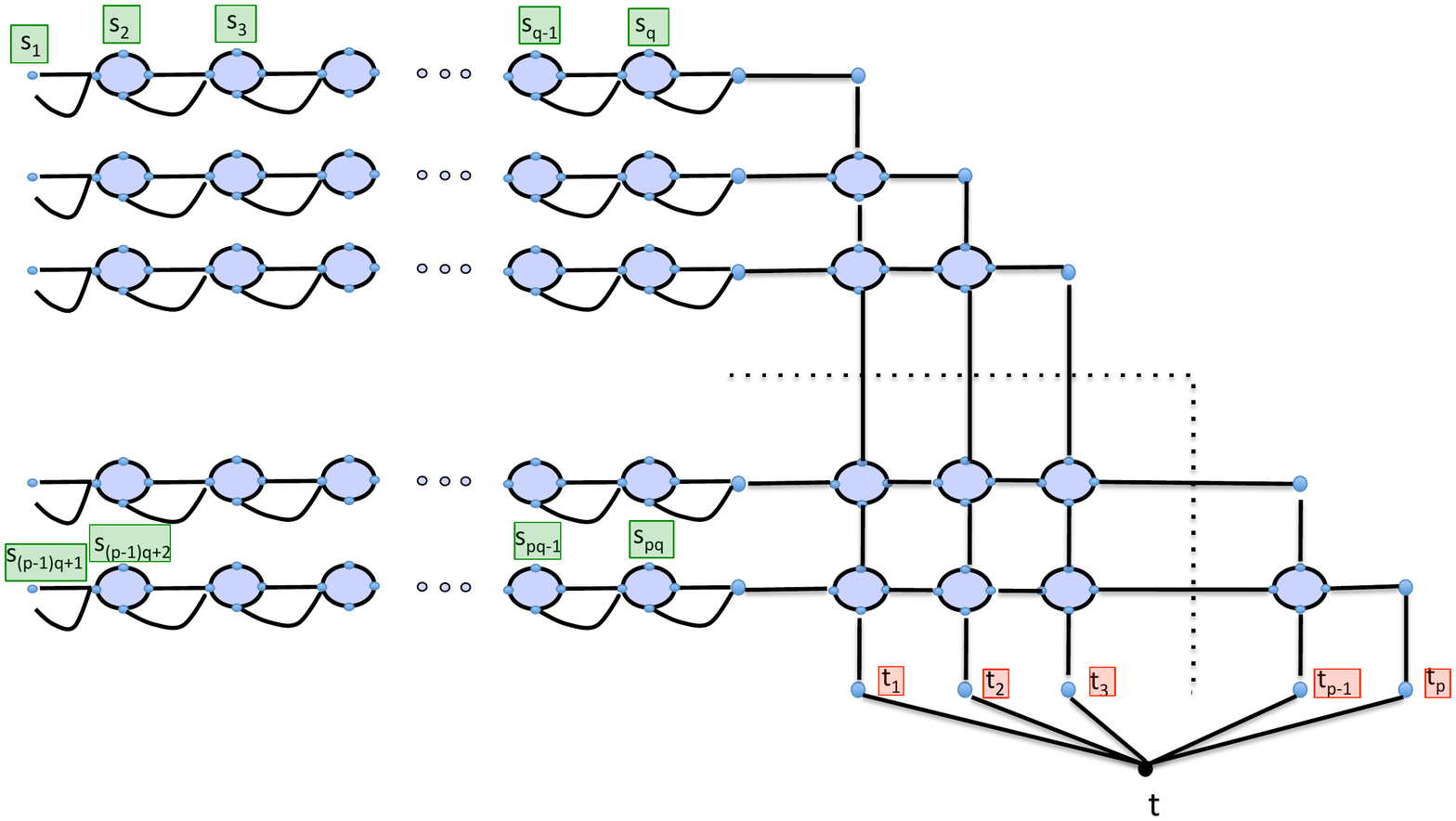}\\
\vspace{-2.5cm}
\caption{\label{fig:gridline}}
\end{center}
\end{figure}

 Now take the $i^{th}$ Azar-Regev path, for each $i=1,2 \ldots p$. The path
 contains $q$ supply nodes with demands of sizes $2^{(i-1)q},2^{(i-1)q+1},\ldots , 2^{iq-1}$. (Supply node $s_j$ has
 demand $2^{j-1}$.) Therefore the total demand on path $i$ is $\tau_i := 2^{(i-1)q}(2^q-1)<2^{iq}$. The key point is that
{\em the total demand of path $i$ is less than the smallest demand in path $i+1$}.
Note that we have $pq$ demands, and thus demand sizes from $2^{0}$ up to $2^{pq-1}$.
Consequently the demand spread is $2^{pq-1}$. We set $u_{\min}=d_{\min}$ and thus
$$pq-1 = \log\left(\frac{d_{max}}{d_{min}}\right) = \log\left(\frac{d_{max}}{u_{min}}\right)$$

 It remains to prescribe capacities to the edges of the half-grid.
To do this every edge in $i$th canonical hooked path has capacity $\tau_i$ (not $\capp{i}$). These capacity assignments,
in turn, induce corresponding capacities in each of the disjoint paths gadgets.
It follows that if the each gadget on the paths and half-grid correspond to a YES-instance gadget then we may
route all $pq$ demands.

Now suppose the gadgets correspond to a NO-instance. It follows that we may route at most one demand along each Azar-Regev path.
But, by our choice of demand values, any demand on the $i$th path is too large to fit into any column $j < i$ in the half-grid.
Hence we have the same conditions required in Theorem~\ref{thm:hard} to show that at most one demand in total can feasibly route.
It follows that we cannot approximate the cardinality objective to better than a factor $pq$.

Note that the construction contains at most $m=O((qp+ p^2)\cdot |E(H)|)$ edges, where $H$ is the size of the $2$-disjoint paths instance.
Now we select $p$ and $q$ such that $q\ge p$ and $pq\ge |E(H)|^{\frac{1}{\epsilon}}$.
Then, for some constant $C$, we have
\begin{eqnarray*}
C\cdot m^{.\frac12-\epsilon}\cdot \sqrt{\log(\frac{d_{max}}{d_{min}})} &=& m^{\frac12-\epsilon} \cdot \sqrt{\log(\frac{d_{max}}{d_{min}})} \\
&\le & \sqrt{pq}\cdot \sqrt{pq}\\
&=& pq
\end{eqnarray*}
Therefore, since we cannot approximate to within $pq$, we cannot
approximate the cardinality objective to better than a factor $O(\sqrt{\log (\frac{d_{max}}{u_{min}})}\cdot m^{\frac12-\epsilon})$. \qed

\section{Upper Bounds for Flows with Arbitrary Demands}\label{sec:upper}

In this section we present upper bounds for maximum flow problems with arbitrary demands.

\subsection{Unsplittable Flow with Arbitrary Demands}\label{sec:upper-unsplittable}

One natural approach for the cardinality unsplittable flow problem is used repeatedly in the
literature  (even for general multiflows). Group the demands
into at most $O(\log \frac{d_{max}}{d_{min}}) \geq \log (\frac{d_{max}}{u_{min}})$ bins, and then consider
each bin separately. This approach can also be applied to the throughput objective (and to the
more general profit-maximisation model).
This immediately incurs a lost factor relating to the number of bins, and this
feels slightly artificial. In fact, given the no-bottleneck assumption regime, there is no need to lose this extra factor:
Baveja et al~\cite{baveja2000approximation} gave
an $O(\sqrt{m})$ approximation for profit-maximisation when $d_{max} \leq u_{min}$.
On the other hand, our lower bound in Theorem~\ref{thm:harder} shows that if $d_{max} > u_{min}$  we do need to
lose some factor dependent on $d_{max}$. But how large does this need to be?
The current best upper bound is $O(\log (\frac{d_{max}}{u_{min}})\cdot \sqrt{m \log m})$ by Guruswami et al.~\cite{guruswami2003near},
and this works for the general profit-maximisation model.\footnote{Actually, they state the bound as $\log^{3/2}m$ because
exponential size demands are not considered in that paper.}
For the cardinality and throughput objectives, however, we can obtain a better upper bound.
The proof combines analyses from \cite{baveja2000approximation} and \cite{kolliopoulos2004approximating}
(which focus on the no-bottleneck assumption case). We emphasize that the following theorem
applies to all multiflow problems not just the single-sink case.\\



{{\noindent\bf Theorem~\ref{thm:upper}.} \itshape
There is an $O(\sqrt{m}\log (\frac{d_{max}}{u_{min}}))$ approximation algorithm
for cardinality unsplittable flow and an $O(\sqrt{m}\log n)$ approximation algorithm
for throughput unsplittable flow, in both directed and undirected graphs. \\
}

\noindent{\bf Proof.}
We apply a result from \cite{guruswami2003near} which shows that for cardinality
unsplittable flow, with $d_{max} \leq \Delta d_{min}$,
the greedy algorithm yields a $O(\Delta \sqrt{m})$ approximation.  Their proof is a technical extension
of the greedy analysis of Kolliopoulos and Stein \cite{kolliopoulos2004approximating}.
We first find an approximation for the sub-instance consisting of the demands at most $u_{min}$. This
satisfies the no-bottleneck assumption and an $O(\sqrt{m})$-approximation is known for general profits \cite{baveja2000approximation}.
Now, either this sub-instance gives half the optimal profits, or we focus on demands of at least $u_{min}$.
In the remaining demands, by losing a $\log (\frac{d_{max}}{u_{min}})$ factor, we may assume $d_{max} \leq \Delta d_{min}$, for some $\Delta=O(1)$.
The greedy algorithm above then gives the desired guarantee for the cardinality problem.
The same approach applies for the throughput objective, since all demands within the same bin
have values within a constant factor of each other. Moreover, we require only $\log n$ bins
as demands of  at most $\frac{d_{max}}{n}$ may be discarded as they are not necessary for
obtaining high throughput.
\qed




\vspace*{.3cm}
As alluded to earlier, this  upper bound  is not completely satisfactory as pointed out in \cite{chekuri2003edge}. Namely,
all of the lower bound instances have a linear number of edges $m=O(n)$.
Therefore, it is possible that there exist upper bounds dependent on $\sqrt{n}$.
Indeed, for the special case of \medp \ in undirected graphs and directed acyclic graphs
$O(\sqrt{n})$-approximations have been developed  \cite{chekuri2006n,nguyen2007disjoint}.
Such an upper bound is
not known for general directed \medp \ however;
the current best approximation is $\min\{\sqrt{m},n^{2/3}\}$.

\subsection{Priority Flow with Arbitrary Demands}\label{sec:upper-priority}
Next we show that the lower bound for the maximum priority flow problem is tight.
\begin{thm}
Consider an instance of the maximum priority flow problem with $k$ priority classes.
There is a polytime algorithm that approximates the maximum flow to within a factor of $O(\min\{k,\sqrt{m}\})$.
\end{thm}
\begin{proof}
 First suppose that $k \leq \sqrt{m}$. Then for each class $i$,
we may find the optimal priority flow by solving a maximum flow problem in the subgraph induced by all edges of
priority $i$ or better. This yields a $k$-approximation.
Next consider the case where $\sqrt{m} < k$. Then we may apply Lemma \ref{lem:greedy-priority}, below,
which implies that the greedy algorithm yields
a $O(\sqrt{m})$-approximation. The theorem follows.
\qed
\end{proof}

 The following proof for uncapacitated networks follows ideas from the greedy analysis of Kleinberg \cite{Kleinberg96},
 and Kolliopoulos and Stein \cite{kolliopoulos2004approximating}. One may also design an $O(\sqrt{m})$-approximation
 for general edge capacities using more intricate ideas from \cite{guruswami2003near}; we omit the details.
\begin{lemma}\label{lem:greedy-priority}
A greedy algorithm yields a $O(\sqrt{m})$-approximation to the  maximum priority flow problem.
\end{lemma}
\begin{proof}
We now run the greedy algorithm as follows. On each iteration, we find the demand $s_i$ which has a shortest feasible path
 in the residual graph. Let $P_i$ be the associated path, and delete
its edges. Let the greedy solution have cardinality $t$.
 Let $\mathcal{O}$ be the optimal maximum priority flow and
let ${\cal Q}$ be those demands which are satisfied in some optimal
solution but not by the greedy algorithm.
We aim to upper bound the size of $\mathcal{Q}$.

Let $Q$ be a path used in the optimal solution satisfying some demand in ${\cal Q}$.
Consider any edge $e$ and the greedy path using it.
 We say that $P_i$
{\em blocks} an optimal path $Q$ if $i$ is the least index such
that $P_i$ and $Q$ share a common edge $e$.
Clearly such an $i$ exists or else we could still route on $Q$.

Let $l_i$ denote the length of  $P_i$.
Let $k_i$ denote the number of optimal paths  (corresponding to
demands in ${\cal Q}$ ) that are blocked by $P_i$. It follows that
$k_i \le l_i$. But, by the definition of the greedy algorithm,
we have that each such blocked path must have length at least $l_i$
 at the time when $P_i$ was packed.
Hence it used up at least $l_i  \ge k_i$ units of capacity in the optimal solution.
Therefore the total capacity used by the unsatisfied demands from the
optimal solution is at least $\sum_{i=1}^t k_i^2$.
As the total capacity is at most $m$ we obtain

\begin{equation}
\label{bounding}
\frac{(\sum_{i=1}^t k_i)^2}{t} \leq
\sum_{i=1}^t k_i^2 \leq m
\end{equation}

\noindent
where the first inequality is by the Chebyshev Sum Inequality.
Since $\sum_{i} k_i = |{\cal Q}| = |{\cal O}| -t $, we obtain
$\frac{(|\mathcal{O}|-t)^2}{t} \leq m$. One may verify that if $t < \frac{|\mathcal{O}|}{\sqrt{m}}$ then this inequality implies
$|\mathcal{O}| = O(\sqrt{m})$ and, so, routing a single demand yields the desired approximation.
\qed
\end{proof}

\section{Conclusion}
 It would be interesting to improve the upper bound in Theorem~\ref{thm:upper} to be in terms of $\sqrt{n}$
 rather than $\sqrt{m}$.
Resolving the discrepancy with Theorem~\ref{thm:harder}
between $\sqrt{\log(\frac{d_{max}}{u_{min}})}$ and $\log(\frac{d_{max}}{u_{min}})$ would also
clarify the complete picture.

\ \\
\noindent{\bf Acknowledgments.}
The authors thank Guyslain Naves for his careful reading and precise and helpful comments.
The authors gratefully acknowledge support from the NSERC Discovery Grant Program. 

\bibliography{conf}

\end{document}